\documentclass[a4paper]{article}
\usepackage{amsfonts,amssymb,amsmath,amsthm,cite}
\usepackage{ucs} 
\usepackage[utf8x]{inputenc}
\usepackage{graphicx}
\usepackage[english]{babel}
\usepackage{slashed}
\usepackage{textcomp}
\usepackage{bm}
\usepackage{titlesec}

\evensidemargin=\oddsidemargin
\newtheorem{lemma}{Lemma}

\begin{document}

\vspace{10mm}
\begin{center}
	\large{\textbf{On a criterion for a cutoff regularization}\\ 
		\textbf{in the coordinate representation}}
\end{center}
\vspace{2mm}
\begin{center}
	\large{\textbf{Aleksandr V. Ivanov}}
\end{center}
\begin{center}
	St. Petersburg Department of Steklov Mathematical Institute of Russian Academy of Sciences,\\ 
	27 Fontanka, St. Petersburg 191023, Russia
\end{center}
\begin{center}
	Leonhard Euler International Mathematical Institute in Saint Petersburg,\\ 
	10 Pesochnaya nab., St. Petersburg 197022, Russia
\end{center}
\begin{center}
	E-mail: regul1@mail.ru
\end{center}
\vspace{2mm}
\begin{flushright}
	\large{\textbf{\textit{To the 90-th anniversary of L.D.Faddeev}}}
\end{flushright}
\vspace{10mm}

\textbf{Abstract.} The paper discusses an applicability criterion for a cutoff regularization in the coordinate representation in the Euclidean space with a dimension larger than two. It is shown that the set of functions satisfying the criterion is not empty. As an example, an explicit function is presented. It is proved by explicit construction that there are functions satisfying the criterion in a stronger formulation.

\vspace{2mm}
\textbf{Key words and phrases:} cutoff regularization, Green's function, fundamental solution, deformation, coordinate representation.

\newpage

\section{Introduction}
\label{31:sec:int}
Divergent integrals arise in various models of quantum field theory \cite{10,7}. As a rule, this is due to the fact that generalized functions \cite{Gelfand-1964}, which should be considered on a certain test class, act on other generalized functions. This leads to the appearance of non-integrable densities. To work with such objects, additional regularization is necessary, the choice of which depends on the symmetries of the model and significantly affects the process of further research.

In \cite{34}, a cutoff regularization in the coordinate representation was proposed. Later, it was improved \cite{Ivanov-2022} and successfully applied to a number of models \cite{Ivanov-Kharuk-2020,Ivanov-Kharuk-20222,Ivanov-Kharuk-2023,Iv-2024-1}. In this paper, we formulate an applicability criterion for such regularization, prove its feasibility, and give some specific examples.

Let $n\in\mathbb{N}$ and $n>2$. Consider the Euclidean space $\mathbb{R}^n$ with the standard scalar product of vectors $(\,\cdot\,,\cdot\,)$. Let us define the Laplace operator $A_n(x)$ and a fundamental solution $G_n(x)$, which in Cartesian coordinates have the form
\begin{equation*}\label{31-1}
A_n(x)=-\sum_{k=1}^n\partial_{x_k}^2,\,\,\,
G_n(x)=\frac{|x|^{2-n}}{(n-2)S_{n-1}},\,\,\,
S_{n-1}=\frac{2\pi^{n/2}}{\Gamma(n/2)}.
\end{equation*}
It is $G_n(\cdot)$ that is usually used as the main approximation of Green's functions, and therefore the study of its properties is an important task. It is clear that the last objects solve the equation $A_n(x)G_n(x-y)=\delta(x-y)$ in the sense of generalized functions on the Schwartz class $S(\mathbb{R}^n)$.

The regularization mentioned above consists in a deformation of the following form
\begin{align}\label{31-2}
G_n^{\phantom{1}}(x)\xrightarrow{\mbox{\footnotesize{reg.}}}
G_n^{\Lambda,\mathbf{f}}(x)=&\,\frac{\Lambda^{n-2}}{(n-2)S_{n-1}}\mathbf{f}\big(|x|^2\Lambda^2\big)+
\frac{1}{(n-2)S_{n-1}}
  \begin{cases} 
\,\,\Lambda^{n-2}, & \mbox{if}\,\,\,|x|\leqslant1/\Lambda;\\
|x|^{2-n}, &\mbox{if}\,\,\,|x|>1/\Lambda,
\end{cases}\\\nonumber
=&\,\frac{\Lambda^{n-2}}{(n-2)S_{n-1}}\mathbf{f}\big(|x|^2\Lambda^2\big)+
G_n^{\Lambda,\mathbf{0}}(x),
\end{align}
where $\Lambda$ is a regularizing parameter, and $\mathbf{f}(\cdot)\in C\big([0,+\infty),\mathbb{R}\big)$ is an auxiliary deforming function satisfying the properties
\begin{equation*}\label{31-3}
	\mathrm{supp}\big(\mathbf{f}(\cdot)\big)\subset[0,1]\,\,\,\mbox{and}\,\,\,
	A_n(x)\Lambda^{n-2}\mathbf{f}\big(|x-y|^2\Lambda^2\big)
	\xrightarrow{\Lambda\to+\infty}0.
\end{equation*}
For this type of deformation, the transition  $A_n^{\phantom{1}}(x)G_n^{\Lambda,\mathbf{f}}(x-y)\to0$ is valid when removing the regularization $\Lambda\to+\infty$  in the sense of generalized functions on $S(\mathbb{R}^n)$. It follows from the construction that the deformation is performed only in the closed ball  $\mathrm{B}_{1/\Lambda}$ of radius $1/\Lambda$ with the center at the origin, therefore, for $|x|>1/\Lambda$, the equality $G_n^{\phantom{1}}(x)=
G_n^{\Lambda,\mathbf{f}}(x)$ holds. In simple words, we can say that when regularization is introduced, we cut off the increasing function in the region $\mathrm{B}_{1/\Lambda}$.

In the study of quantum field models, the regularized fundamental solution is the kernel of the operator in a quadratic form, the value of which should not be negative. In other words, the Fourier transform of the function $G_n^{\Lambda,\mathbf{f}}(\cdot)$ must not take negative values for all argument values and all $\Lambda>N>0$ for some fixed number $N$. In mathematical language, such relationship can be formulated as follows
\begin{equation}\label{31-5}
\hat{G}_n^{\Lambda,\mathbf{f}}(y)=
\int_{\mathbb{R}^n}\mathrm{d}^nx\,e^{i(y,x)}
G_n^{\Lambda,\mathbf{f}}(x)\geqslant0,
\end{equation}
for all $y\in\mathbb{R}^n$ and $\Lambda>N>0$. The latter relation is an applicability criterion for the regularization. It defines a class of valid functions $\mathbf{f}(\cdot)$ for which the spectral density does not have negative values.

In the short form, the main results of the work can be summarized as follows:
\begin{itemize}
	\item Lemma \ref{31-le1} contains the formulation of the criterion with respect to the deforming function;
	\item Lemma \ref{31-le2} shows that the set of functions satisfying the criterion is not empty;
	\item Lemma \ref{31-le3} gives an explicit form of a function satisfying the criterion;
	\item in Lemma \ref{31-le4} a function satisfying the criterion in a strict formulation is constructed.
\end{itemize}

\section{Results}
\label{31:sec:lem}

\begin{lemma}\label{31-le1}
Taking into account all the above, the applicability criterion \eqref{31-5} can be equivalently represented by the condition
\begin{equation}\label{31-6}
\frac{s^2}{n-2}\int_0^1\mathrm{d}t\,t^{n-1}\rho_n(ts)\mathbf{f}\big(t^2\big)+\rho_n(s)
\geqslant0\,\,\,\mbox{for all}\,\,\,s\geqslant0,
\end{equation}
where
\begin{equation}\label{31-7}
\rho_n(s)=\Gamma(n/2)(s/2)^{1-n/2}J_{n/2-1}(s),
\end{equation}
and $J_{n/2-1}(\cdot)$ is the Bessel function of the first kind.
\end{lemma}
\begin{proof} Substitute the explicit form of the deformed fundamental solution \eqref{31-2} into  inequality \eqref{31-5} and use the relations, see formulas (20) and (30) in \cite{Ivanov-2022} and Theorem 4.15 in \cite{31-1},
\begin{equation}\label{31-8}
\int_{\mathbb{R}^n}\mathrm{d}^nx\,G_n^{\Lambda,\mathbf{0}}(x)e^{i(x,y)}=\frac{\rho_n(|y|/\Lambda)}{|y|^2}\,\,\,\mbox{and}\,\,\,
\rho_n(|y|)=\frac{1}{S_{n-1}}\int_{\mathrm{S}^{n-1}}\mathrm{d}^{n-1}\sigma(\hat{x})\,e^{i(\hat{x},y)},
\end{equation}
where, in the last equality, integration over the unit sphere $\mathrm{S}^{n-1}$ with the center at the origin with the standard measure was used, $\hat{x}=x/|x|$. Next, using the fact that $|y|\geqslant0$ and $\Lambda>N>0$, we can go to the parameter $s=|y|/\Lambda$, from which the final form \eqref{31-6} follows.
\end{proof}

\begin{lemma}\label{31-le2}
Let $k\in\mathbb{N}\setminus\{0\}$. Consider a set of multi-indices $\{\alpha_i\}_{i=1}^k$, the components of which are positive numbers from the range $(0,1/2]$ and satisfy the relation
\begin{equation}\label{31-9}
2\sum_{j=1}^{\dim(\alpha_i)}(\alpha_i)_j\leqslant1
\,\,\,\mbox{for all}\,\,\,i\in\{1,\ldots,k\}.
\end{equation}
Next, for each multi-index $\alpha_i$, we define an integral operator of the form
\begin{equation}\label{31-10}
\mathrm{H}_{\alpha_i}^{\Lambda}:\,\,\,g(x)\to\mathrm{H}_{\alpha_i}^{\Lambda}(g)(x)=
\Bigg(\prod_{j=1}^{2\dim(\alpha_i)}\int_{\mathrm{S}^{n-1}}\frac{\mathrm{d}^{n-1}\sigma(\hat{x}_j)}{S_{n-1}}\Bigg)\,
g\Bigg(x+\Lambda^{-1}\sum_{m=1}^{\dim(\alpha_i)}\big(\hat{x}_{2m}+\hat{x}_{2m-1}\big)(\alpha_i)_m\Bigg),
\end{equation}
where $g(\cdot)$ is an auxiliary function for which integrals exist. We also define a set of positive numbers $\{\kappa_i\}_{i=1}^k$, such that $\kappa_1+\ldots+\kappa_k=1$. Then the function 
\begin{equation}\label{31-11}
\sum_{i=1}^k\kappa_i\mathrm{H}_{\alpha_i}^{\Lambda}(G_n)(x)
\end{equation}
has a deformation $G^{\Lambda,\mathbf{f}_n}_n(x)$ of the form \eqref{31-2}, and the deforming function $\mathbf{f}_n(\cdot)\in C\big([0,+\infty),\mathbb{R}\big)$ can be written out as follows
\begin{equation}\label{31-12}
\mathbf{f}_n(s)=(n-2)S_{n-1}\sum_{i=1}^k\kappa_i\mathrm{H}_{\alpha_i}^{1}(G_n)(\sqrt{s}\hat{x})-
\begin{cases} 
	\,\,\,\,\,\,1, & \mbox{if}\,\,\,s\leqslant1;\\
	s^{1-n/2}, &\mbox{if}\,\,\,s>1.
\end{cases}
\end{equation}
\end{lemma}
\begin{proof} Note that if each function from the set satisfies the criterion \eqref{31-6}, then their convex linear combination also satisfies it. Therefore, without limiting generality, it is sufficient to consider only the case of $\alpha_i$ for a fixed index $i$. 
	
Note that by definition the operator \eqref{31-10} is a multiple homogenization ($2\dim(\alpha_i)$ times). At the same time, it follows from formulas \eqref{31-8} that even the first homogenization of the function $G_n^{\phantom{1}}(x)$ with a radius $r>0$ is equal to
\begin{equation*}\label{31-13}
\frac{1}{S_{n-1}}\int_{\mathrm{S}^{n-1}}\mathrm{d}^{n-1}\sigma(\hat{x})\,G_n^{\phantom{1}}(y+r\hat{x})=G_n^{1/r,\mathbf{0}}(y)
\end{equation*}
and is a bounded function. Therefore, the operator $\mathrm{H}_{\alpha_i}^\Lambda$ regularizes the fundamental solution. Let us make sure that the criterion from \eqref{31-5} is satisfied with such deformation. To do this, we apply the Fourier transform and use formulas \eqref{31-8}, then we get the following non-negative density
\begin{equation*}\label{31-14}
\int_{\mathbb{R}^n}\mathrm{d}^nx\,e^{i(y,x)}
\mathrm{H}_{\alpha_i}^{\Lambda}(G_n)(x)=\frac{1}{|y|^2}\prod_{j=1}^{\dim(\alpha_i)}
\Big(\rho_n\big(|y|(\alpha_i)_j/\Lambda\big)\Big)^2\geqslant0.
\end{equation*}
Next, check that the function $\mathrm{H}_{\alpha_i}^{\Lambda}(G_n)(x)$ fits into the representation \eqref{31-2}. To do this, note that if the support of some auxiliary function $g(\cdot)$ lies in a ball $\mathrm{B}_{r_1}$ of a radius $r_1$, then for the support of homogenized function with a radius $r_2$ the following relation holds
\begin{equation*}\label{31-15}
\mathrm{supp}\Bigg(\int_{\mathrm{S}^{n-1}}\frac{\mathrm{d}^{n-1}\sigma(\hat{y})}{S_{n-1}}\,
	g(\,\,\cdot\,+r_2\hat{y})\Bigg)\subset\mathrm{B}_{r_1+r_2}.
\end{equation*}
Therefore, if at some step we get the function $\tilde{G}(x)$, which can be obtained from $G_n(x)$ by deforming in the ball $\mathrm{B}_{r_1}$, that is, $\tilde{G}(x)=G_n(x)$ in $\mathbb{R}^n\setminus\mathrm{B}_{r_1}$, then additional homogenization with a radius $r_2$ leads to the equality
\begin{equation*}\label{31-16}
\int_{\mathrm{S}^{n-1}}\frac{\mathrm{d}^{n-1}\sigma(\hat{y})}{S_{n-1}}\,
	\tilde{G}(x+r_2\hat{y})=
\int_{\mathrm{S}^{n-1}}\frac{\mathrm{d}^{n-1}\sigma(\hat{y})}{S_{n-1}}\,
\Big(\tilde{G}(x+r_2\hat{y})-G_n(x+r_2\hat{y})\Big)+G_n^{1/r_2,\mathbf{0}}=
G_n(x)
\end{equation*}
for all $x\in\mathbb{R}^n\setminus\mathrm{B}_{r_1+r_2}$. This procedure can be applied step by step to the operator from \eqref{31-10}. Thus, given relation \eqref{31-9}, it can be argued that $\mathrm{H}_{\alpha_i}^{\Lambda}(G_n)(x)=G_n(x)$ for all $x\in\mathbb{R}^n\setminus\mathrm{B}_{1/\Lambda}$, including for $|x|=1/\Lambda$ due to the continuity of the homogenized function. This means that the function $\mathrm{H}_{\alpha_i}^{\Lambda}(G_n)(\cdot)$ has the deformation of the form \eqref{31-2}. 

The last formula \eqref{31-12} is obtained by reversing formula \eqref{31-2} with additional scaling of the variables. Note that the resulting function is spherically symmetric and in fact does not depend on the choice of the unit vector $\hat{x}$.
\end{proof}

\begin{lemma}\label{31-le3} Let the assumptions of Lemma \ref{31-le2} be satisfied. Let us take $k=1$, $\dim(\alpha_1)=1$, and $(\alpha_1)_1=1/2$, then the function from \eqref{31-12} is written out explicitly
\begin{equation}\label{31-17}
\mathbf{f}_n(s)=2^{n-2}-1-\frac{2^{n-1}\sqrt{s}S_{n-2}}{S_{n-1}}{}_2F_{1}\bigg(\frac{1}{2},\frac{3-n}{2};\frac{3}{2};s\bigg)+\frac{2^{n-1}\sqrt{s}S_{n-2}}{(n-1)S_{n-1}}{}_2F_{1}\bigg(\frac{3-n}{2},\frac{n-1}{2};\frac{n+1}{2};s\bigg)
\end{equation}
for $s\in[0,1]$, and $\mathbf{f}_n(s)=0$ for $s>1$. Here ${}_2F_{1}$ is the hypergeometric function. In particular, in the dimensions $n\in\{3,4,5,6\}$ in the domain $s\in[0,1]$ we have the following explicit form
\begin{equation*}\label{31-18}
\mathbf{f}_3(s)=1-\sqrt{s},
\end{equation*}
\begin{equation*}\label{31-19}
\mathbf{f}_4(s)=3-\frac{4s+2}{\pi}\bigg(\frac{1-s}{s}\bigg)^{1/2}+\frac{2}{\pi}\bigg(\frac{1}{s}-4\bigg)\mathrm{arcsin}\big(\sqrt{s}\big),
\end{equation*}
\begin{equation*}\label{31-20}
\mathbf{f}_5(s)=7-9\sqrt{s}+2s^{3/2},
\end{equation*}
\begin{equation*}\label{31-21}
\mathbf{f}_6(s)=15+\frac{2}{3\pi}\big(16s^3-56s^2-2s-3\big)\bigg(\frac{1-s}{s^3}\bigg)^{1/2}+
\frac{2}{\pi}\bigg(\frac{1}{s^2}-16\bigg)\mathrm{arcsin}\big(\sqrt{s}\big).
\end{equation*}
The functions are shown in Fig. \ref{31-pic}.
\end{lemma}
\begin{proof} Let us apply the Laplace operator to the function \eqref{31-11} under the condition $\Lambda=1$, then we can write out the following relation with the $\delta$-function in the $n$-dimensional space
\begin{equation}\label{31-22}
A_n(x)\mathrm{H}_{\alpha_1}^1(G_n)(x)=
\int_{\mathrm{S}^{n-1}}\frac{\mathrm{d}^{n-1}\sigma(\hat{y})}{S_{n-1}}
\int_{\mathrm{S}^{n-1}}\frac{\mathrm{d}^{n-1}\sigma(\hat{z})}{S_{n-1}}\,
\delta(x+\hat{y}/2+\hat{z}/2).
\end{equation}
Further, we use the formula
\begin{equation}\label{31-23}
\int_{\mathrm{S}^{n-1}}\mathrm{d}^{n-1}\sigma(\hat{z})\,
\delta(x+\hat{y}/2+\hat{z}/2)=2^{n-1}\delta(|x+y/2|-1/2),
\end{equation}
where the right hand side already contains the $\delta$-function in the one-dimensional space. Then, in the remaining integration over the sphere, we go to the following spherical coordinates
\begin{equation*}\label{31-24}
\hat{y}_1=\cos(\phi_1),\,\,\,
\hat{y}_2=\sin(\phi_1)\cos(\phi_2),\,\,\,\ldots\,\,\,
\hat{y}_{n-1}=\sin(\phi_1)\cdot\ldots\cdot\sin(\phi_{n-2})\cos(\phi_{n-1}),
\end{equation*}
\begin{equation*}\label{31-25}
\mathrm{d}^{n-1}\sigma(\hat{y})=\sin^{n-2}(\phi_1)\sin^{n-3}(\phi_2)\cdot\ldots\cdot\sin(\phi_{n-2})\,\mathrm{d}\phi_1\ldots\mathrm{d}\phi_{n-1},
\end{equation*}
where $\phi_i\in[0,\pi]$ for $i\in\{1,\ldots,n-2\}$, and $\phi_{n-1}\in[0,2\pi)$.
For convenience, we choose them in such a way that the relation is fulfilled
\begin{equation*}\label{31-26}
(x,\hat{y})=|x|\cos(\phi_1).
\end{equation*}
Then, taking into account \eqref{31-23}, the right hand side of \eqref{31-22} is rewritten as
\begin{equation*}\label{31-27}
\frac{2^{n-1}}{S_{n-1}^2}
\int_{\mathrm{S}^{n-1}}\mathrm{d}^{n-1}\sigma(\hat{y})\,
\delta(|x+y/2|-1/2)=\frac{2^{n-1}S_{n-2}}{S_{n-1}^2}
\int_0^\pi\mathrm{d}\phi_1\,\sin^{n-2}(\phi_1)
\delta\bigg(\sqrt{|x|^2+|x|\cos(\phi_1)+\frac{1}{4}}-\frac{1}{2}\bigg).
\end{equation*}
The last integral is explicitly calculated, and the intermediate answer takes the form
\begin{equation}\label{31-28}
A_n(x)\mathrm{H}_{\alpha_1}^1(G_n)(x)=\frac{2^{n-1}S_{n-2}}{S_{n-1}^2}
\begin{cases}
|x|^{-1}\big(1-|x|^2\big)^{\frac{(n-3)}{2}},& \mbox{if}\,\,\,|x|\leqslant1;\\
~~~~~~~~~~~~0,& \mbox{if}\,\,\,|x|>1.
\end{cases}
\end{equation}
Using the fact that the latter function depends only on $|x|$, the operator $A_n(x)$ can be represented as $-|x|^{1-n}\partial_{|x|}|x|^{n-1}\partial_{|x|}$. Then, integrating the function \eqref{31-28}, taking into account the condition that the result must coincide with $G_n(x)$ for $|x|>1$, and substituting it into formula \eqref{31-12}, we obtain the declared relation from \eqref{31-17}. Special cases follow from the definition of the hypergeometric function.
\end{proof}
\begin{figure}[h]
	\center{\includegraphics[width=0.4\linewidth]{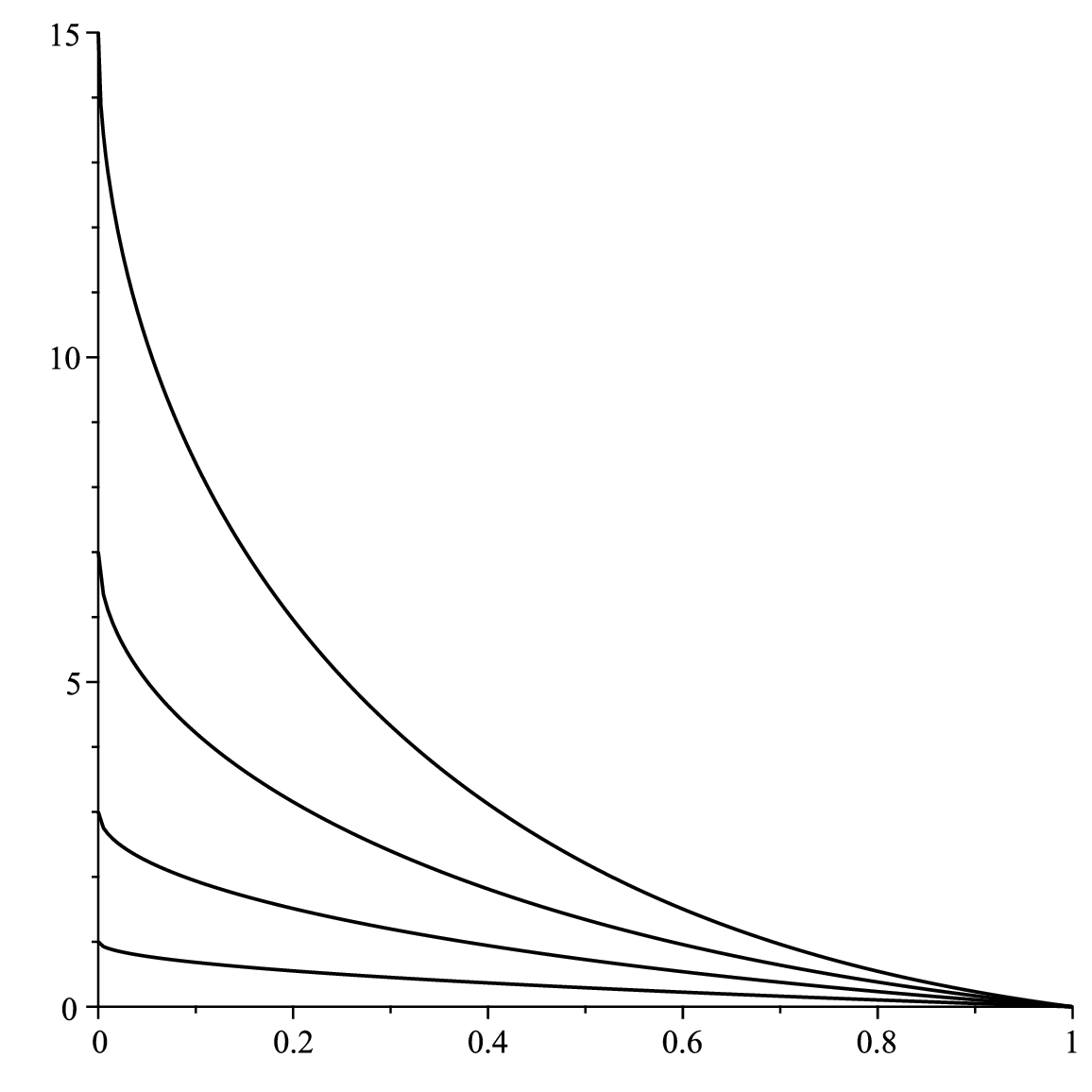}}
	\caption{The function $\mathbf{f}_n(s)$ for $n\in\{3,4,5,6\}$ and $s\in[0,1]$. In the picture $\mathbf{f}_3(0)<\mathbf{f}_4(0)<\mathbf{f}_5(0)<\mathbf{f}_6(0)$.}
	\label{31-pic}
\end{figure}

\begin{lemma}\label{31-le4}  Let the assumptions of Lemma \ref{31-le3} be satisfied.
Let us consider two sets of positive numbers $\{\kappa_i\}_{i=1}^{+\infty}$ and $\{r_i\}_{i=1}^{+\infty}$, satisfying the following conditions
\begin{equation}\label{31-29}
\sum_{i=1}^{+\infty}\kappa_i=1,\,\,\,
\lim_{i\to+\infty}(r_i)=0,\,\,\,
\sum_{i=1}^{+\infty}\kappa_i^{\phantom{1}}r_i^{-n}=\kappa<+\infty.
\end{equation}
Let $r=\max_{i\geqslant1}(r_i)$. Next, we connect the deformed fundamental solution $G_n^{\Lambda/r_i,\mathbf{f}_n}$ to each $r_i$, where $\mathbf{f}_n(\cdot)$ is from \eqref{31-17}. Then the function
\begin{equation}\label{31-30}
\sum_{i=1}^{+\infty}\kappa_iG_n^{\Lambda/r_i,\mathbf{f}_n}
\end{equation}
is a deformation $G_n^{\Lambda/r,\mathbf{\tilde{f}}_n}$ of the form \eqref{31-2} for the solution $G_n^{\phantom{1}}(x)$, and the corresponding deforming function $\mathbf{\tilde{f}}_n(\cdot)$ solves the inequality from \eqref{31-6} in the strong formulation, that is, with the sign $>$ instead of $\geqslant$.
\end{lemma}
\noindent\textbf{Remark.} As examples of the sequences, we can take $\kappa_i=2^{-i}$ and $r_i=2^{-i/(2n)}$.
\begin{proof} Note that each individual summand in the sum \eqref{31-30} has the deformation of the form \eqref{31-2}. Moreover, the function $G_n^{\Lambda/r_i,\mathbf{f}_n}$ coincides with $G_n^{\phantom{1}}(x)$ for all $i\geqslant1$ and $x\in\mathbb{R}^n\setminus\mathrm{B}_{r/\Lambda}$. Therefore, given the first equality from \eqref{31-29}, we make sure that the sum \eqref{31-30} is equal to $G_n^{\phantom{1}}(x)$ in the domain $x\in\mathbb{R}^n\setminus\mathrm{B}_{r/\Lambda}$. Next, check that the sum is finite in the region $\mathrm{B}_{r/\Lambda}$. Let us define the following number
\begin{equation*}\label{31-33}
M=\max_{x\in\mathrm{B}_{1}}\Bigg(
\frac{\big|\mathbf{f}_n\big(|x|^2\big)\big|+1}{(n-2)S_{n-1}}\Bigg).
\end{equation*}
Then, using that the inclusion $\mathrm{B}_{r_i/\Lambda}\subset\mathrm{B}_{r/\Lambda}$ holds for all $i\geqslant1$, each function $G_n^{\Lambda/r_i,\mathbf{f}_n}$ in the ball with the radius $r/\Lambda$ can be estimated as follows
\begin{equation*}\label{31-31}
\big|G_n^{\Lambda/r_i,\mathbf{f}_n}(x)\big|\leqslant
\frac{(\Lambda/r_i)^{n-2}}{(n-2)S_{n-1}}\big|\mathbf{f}_n\big(|x|^2\Lambda^2/r_i^2\big)\big|+
\frac{(\Lambda/r_i)^{n-2}}{(n-2)S_{n-1}}\leqslant(\Lambda/r_i)^{n-2}M.
\end{equation*}
Thus, the following estimation is valid
\begin{equation*}\label{31-32}
\bigg|\sum_{i=1}^{+\infty}\kappa_iG_n^{\Lambda/r_i,\mathbf{f}_n}(x)\bigg|\leqslant
\sum_{i=1}^{+\infty}\kappa_i\Big|G_n^{\Lambda/r_i,\mathbf{f}_n}(x)\Big|\leqslant
\Lambda^{n-2}Mr^2\sum_{i=1}^{+\infty}\kappa_i^{\phantom{1}}r_i^{-n}(r_i/r)^2\leqslant\Lambda^{n-2}Mr^2\kappa,
\end{equation*}
which implies the boundedness and existence of the deformed function. The continuity of the limit function follows from uniform convergence, which can be obtained with the usage of the Weierstrass M-test. Despite the fact that for the last estimate it is sufficient to have the convergence of the series with the terms $\kappa_i^{\phantom{1}}r_i^{2-n}$, the addition in the power ($r_i^{-n}$ instead of $r_i^{2-n}$) ensures the existence of two derivatives.

Let us show that relation \eqref{31-5} holds in the strong form for the function $G_n^{\Lambda/r,\mathbf{\tilde{f}}_n}$. To do this, we apply the Fourier transform
\begin{equation}\label{31-34}
\hat{G}_n^{\Lambda/r,\mathbf{\tilde{f}}_n}(y)=
\sum_{i=1}^{+\infty}\kappa_i\hat{G}_n^{\Lambda/r_i,\mathbf{f}_n}(y)=
\frac{1}{|y|^2}\sum_{i=1}^{+\infty}\kappa_i\Big(\rho_n\big(|y|r_i/(2\Lambda)\big)\Big)^2.
\end{equation}
Next, it must be shown that for any value $y\in\mathbb{R}^n$ and for all $\Lambda>N>0$ for some fixed $N$ the last sum is strictly greater than zero. Note that the function $\rho_n(\cdot)$ is oscillating, see \eqref{31-7}, and it starts with the point $\rho_n(0)=1$. Let the number $\theta$ denote the first zero. Then we note that from the second relation from \eqref{31-29} it follows that there are numbers $r_j$, which are less than an arbitrarily small fixed value. Next, we fix an arbitrary $N>0$. Then, for any $y\in\mathbb{R}^n$, one can find such a small $r_j$ that the inequality $|y|r_j/\Lambda\leqslant\theta$ holds for all $\Lambda>N$. In this case, continuing \eqref{31-34}, we can write
\begin{equation*}\label{31-35}
	\hat{G}_n^{\Lambda/r,\mathbf{\tilde{f}}_n}(y)\geqslant
	\frac{\kappa_jm^2}{|y|^2}>0,\,\,\,\mbox{where}\,\,\,
m=\min_{s\in[0,|y|]}\Big(\rho_n\big(sr_j/(2\Lambda)\big)\Big),
\end{equation*}
from which it follows that the inequality in the strict form is fulfilled.
\end{proof}

\section{Conclusion}

The paper describes the applicability criterion for the cutoff regularization in the coordinate representation in the Euclidean space with dimensions larger than two. Examples of deformations satisfying the criterion are constructed. In addition, the condition has been analyzed in the stronger formulation.

The text purposefully omits the case with dimension $n=2$. This is due to the fact that the logarithmic singularity of the form $\ln(\Lambda/\sigma)$ can participate in the deformation, so the ansatz can contain several terms. Such case is supposed to be studied in a separate paper with relevant examples. The application of regularization with $n=2$ can be found in \cite{Ivanov-Akac}. As an example, it can be noted that in the case of choosing a function in the form
\begin{equation*}
	G_2^{\Lambda,\mathbf{f}}(x)=\frac{1}{4\pi}\mathbf{f}\big(|x|^2\Lambda^2\big)+
	G_2^{\Lambda,\mathbf{0}}(x)
\end{equation*}
with the condition $\hat{G}_2^{\Lambda,\mathbf{0}}(y)=\rho_2\big(|y|/\Lambda\big)/|y|^2$, the results of Lemmas \ref{31-le2} (except for formula \eqref{31-12}) and \ref{31-le4} and formula \eqref{31-28} remain valid for $n=2$.

\vspace{2mm}
\textbf{Acknowledgements.} The work is supported by the Ministry of Science and Higher Education of the Russian Federation, grant 075-15-2022-289, and by the Foundation for the Advancement of Theoretical Physics and Mathematics "BASIS", grant "Young Russian Mathematics".

\vspace{2mm}
The author expresses gratitude to N.V.Kharuk for a careful reading of the text, numerous comments, criticism and suggestions. Additionally, A.V.Ivanov expresses special gratefulness to N.V.Kharuk and K.A.Ivanov for creating comfortable and stimulating conditions for writing the work.

\vspace{2mm}
\textbf{Data availability statement.} Data sharing not applicable to this article as no datasets were generated or analysed during the current study.

\vspace{2mm}
\textbf{Conflict of interest statement.} The author states that there is no conflict of interest.

\end{document}